\documentclass[12pt]{article}

\usepackage{amsmath,amssymb,amsfonts}
\usepackage{algorithm}
\usepackage{algorithmic}
\usepackage{graphicx}
\usepackage{textcomp}
\usepackage{bm}
\usepackage{array}
\usepackage{booktabs}
\usepackage{url}
\usepackage{float}
\usepackage{geometry}
\usepackage{hyperref}
\usepackage{amsthm}
\newtheorem*{proof*}{Proof}
\renewenvironment{proof}{\noindent\textit{Proof.}\;}{\hfill$\square$\par\medskip}

\newcounter{paperalgorithm}
\renewcommand{\thepaperalgorithm}{\arabic{paperalgorithm}}
\newenvironment{paperalgorithm}[2]{%
\par\medskip
\refstepcounter{paperalgorithm}\label{#2}%
\noindent\begin{minipage}{\linewidth}
\hrule\vspace{0.5ex}
\noindent\textbf{Algorithm~\thepaperalgorithm. #1}
\vspace{0.5ex}
\small
\begin{algorithmic}[1]
}{%
\end{algorithmic}
\vspace{0.5ex}\hrule
\end{minipage}
\medskip
}

\newcounter{definition}
\newcounter{remark}
\newcounter{lemma}
\newcounter{theorem}
\newcounter{corollary}

\newenvironment{definition}[1][]{%
\refstepcounter{definition}\par\medskip\noindent\textbf{Definition~\thedefinition.\if\relax\detokenize{#1}\relax\else\ (#1)\fi}\ \itshape
}{\par\medskip\normalfont}

\newenvironment{remark}[1][]{%
\refstepcounter{remark}\par\medskip\noindent\textbf{Remark~\theremark.\if\relax\detokenize{#1}\relax\else\ (#1)\fi}\ \itshape
}{\par\medskip\normalfont}

\newenvironment{lemma}[1][]{%
\refstepcounter{lemma}\par\medskip\noindent\textbf{Lemma~\thelemma.\if\relax\detokenize{#1}\relax\else\ (#1)\fi}\ \itshape
}{\par\medskip\normalfont}

\newenvironment{theorem}[1][]{%
\refstepcounter{theorem}\par\medskip\noindent\textbf{Theorem~\thetheorem.\if\relax\detokenize{#1}\relax\else\ (#1)\fi}\ \itshape
}{\par\medskip\normalfont}

\newenvironment{corollary}[1][]{%
\refstepcounter{corollary}\par\medskip\noindent\textbf{Corollary~\thecorollary.\if\relax\detokenize{#1}\relax\else\ (#1)\fi}\ \itshape
}{\par\medskip\normalfont}

\title{A Unified Constant-Time Switch Rule for Constructing\\
Edge-Disjoint Hamiltonian Cycles in Gaussian Networks}

\author{Bader Albader\\
\small Department of Computer Science, Faculty of Science,\\
\small Kuwait University, Safat 13060, Kuwait\\
\small \texttt{albader@cs.ku.edu.kw}}

\date{}

\begin{document}

\maketitle

\begin{abstract}
Gaussian networks are degree-four symmetric interconnection networks defined over residue classes of Gaussian integers.  Earlier work showed that when the generator $\alpha=a+bi$ satisfies $\gcd(a,b)=1$, the real and imaginary dimensions directly form two edge-disjoint Hamiltonian cycles.  A later construction extended the result to the non-coprime case $\gcd(a,b)=d>1$, but its proof used long node-sequence tables and separate odd/even cases for $d$.  This paper gives a unified closed-form construction that covers both $d=1$ and $d>1$, and also covers both odd and even $d$, without separate case tables.  In the rectangular representation with $d$ rows and $r=(a^2+b^2)/d$ columns, the construction uses a constant-time local switch rule for each $q=1,2,\ldots,d-1$ at column $a_q=q-1$.  Each switch removes two horizontal edges and inserts two vertical edges.  The switched horizontal structure forms the first Hamiltonian cycle, while its edge-complement in the Gaussian network forms the second Hamiltonian cycle.  Thus, the full edge set is partitioned into two edge-disjoint Hamiltonian cycles.  The construction requires $O(d)$ switch-generation time and $O(N)$ time to list the two cycles, where $N=a^2+b^2$.  Exhaustive validation for all $1\leq a\leq b\leq 100$, excluding only the degenerate $N=2$ network, and large-scale validation up to $N=3{,}250{,}000$ confirm the construction.
\end{abstract}


\medskip

\section{Introduction}
\label{sec:introduction}

Gaussian networks are algebraic interconnection networks defined over quotient rings of Gaussian integers.  They are regular degree-four networks whose vertices are residue classes modulo a Gaussian integer $\alpha=a+bi$, and whose links correspond to adding or subtracting the Gaussian units $1$ and $i$ \cite{Martinez2008,FlahiveBose2010}.  Because they are algebraically regular and symmetric, Gaussian networks are useful for modeling toroidal-like interconnection networks with compact routing and wrap-around structure.

Hamiltonian cycles are important in interconnection networks because a Hamiltonian cycle embeds a spanning ring through all processors.  Such rings are useful for broadcasting, all-to-all communication, and load-balanced message circulation.  Multiple edge-disjoint Hamiltonian cycles (EDHCs) are even more useful because they provide independent spanning rings, allowing communication traffic to be distributed and improving tolerance to edge faults \cite{BaeBose2003,BarthRaspaud1994,Hung2010Twisted,Hung2015Augmented,Pai2023FLTQ,Pai2026BCube}.  Recent work continues to use EDHCs as a structural tool for fault-tolerant broadcasting and data-center/interconnection-network design \cite{Pai2023FLTQ,Pai2026BCube,Yang2023Cayley}.

For Gaussian networks, the EDHC problem naturally splits according to
\begin{equation}
    d=\gcd(a,b).
\end{equation}
When $d=1$, the real dimension alone forms a Hamiltonian cycle and the imaginary dimension alone forms another Hamiltonian cycle.  This gives two edge-disjoint Hamiltonian cycles immediately \cite{FlahiveBose2010}.  When $d>1$, however, the real-dimension edges decompose into $d$ disjoint cycles, and the imaginary-dimension edges also decompose into $d$ disjoint cycles.  The non-coprime case therefore requires a splicing mechanism to join these smaller cycles.  The non-coprime case was previously solved by Albader and Bose using explicit node-sequence tables for odd and even $d$ \cite{AlbaderBoseEDHC}.  The present paper is written to be self-contained: all definitions, switch rules, correctness proofs, and algorithms needed for the unified construction are included here.  The earlier construction established existence but was relatively long and difficult to implement directly.

This paper provides a simpler and unified solution.  It shows that one closed-form switching rule handles all cases:
\begin{equation}
    q=1,2,\ldots,d-1,\qquad a_q=q-1.
\end{equation}
For $d=1$, there are no switches and the construction reduces exactly to the classical coprime case.  For $d>1$, the same formula splices the $d$ horizontal cycles into one Hamiltonian cycle and simultaneously splices the complementary vertical cycles into the second Hamiltonian cycle.  Thus, the new construction unifies the earlier $d=1$ result and the later $d>1$ result in one algorithmic framework.  From a graph-theoretic perspective, the result gives an explicit decomposition of the full edge set of every nondegenerate Gaussian network into two Hamiltonian 2-factors.

The main contributions are:
\begin{itemize}
    \item a unified EDHC construction for nondegenerate Gaussian networks that covers $d=1$, $d>1$, odd $d$, and even $d$;
    \item a closed-form constant-time switch rule $a_q=q-1$, replacing the previous odd/even case tables;
    \item complete algorithms for constructing the switch set and traversing the two Hamiltonian cycles;
    \item a proof that the switched graph is one Hamiltonian cycle, that its edge-complement is also one Hamiltonian cycle, and that the two cycles are edge-disjoint;
    \item exhaustive validation for all $1\leq a\leq b\leq 100$ and large-scale validation up to $3{,}250{,}000$ vertices.
\end{itemize}

\section{Related Work and Positioning}
\label{sec:related}

Gaussian networks were introduced and studied as algebraic models of toroidal networks in \cite{Martinez2008}.  The topology of Gaussian and Eisenstein--Jacobi interconnection networks was further developed in \cite{FlahiveBose2010}, including distance, routing, and Hamiltonian-cycle properties.  In particular, when $\gcd(a,b)=1$, the two natural dimensions of $G_{a+bi}$ give two EDHCs.

The non-coprime Gaussian case was later addressed by Albader and Bose \cite{AlbaderBoseEDHC}.  In that case, the $+1$ edges form $d=\gcd(a,b)$ disjoint cycles and the $+i$ edges form another $d$ disjoint cycles.  The earlier non-coprime construction spliced those cycles but required separate odd/even cases and long tables of node sequences.

The present paper differs from both earlier steps.  It does not only cover $d=1$, and it does not only cover $d>1$.  Instead, it gives a single algorithm that covers both.  Table~\ref{tab:positioning} summarizes the difference.

\begin{table}[H]
\caption{Positioning of the present construction.}
\label{tab:positioning}
\centering
\footnotesize
\begin{tabular}{p{0.28\linewidth}|p{0.18\linewidth}|p{0.38\linewidth}}
\hline
Work & Case covered & Construction style\\
\hline
Flahive and Bose \cite{FlahiveBose2010} & $d=1$ & Direct dimension cycles\\
Albader and Bose \cite{AlbaderBoseEDHC} & $d>1$ & Odd/even sequence tables\\
This paper & $d=1$ and $d>1$ & Unified closed-form switch rule\\
\hline
\end{tabular}
\end{table}

Outside Gaussian networks, EDHCs have been studied in several interconnection topologies, including hypercubes and $k$-ary $n$-cubes \cite{BaeBose2003}, butterfly graphs \cite{BarthRaspaud1994}, twisted and locally twisted cubes \cite{Hung2010Twisted,Hung2010Locally}, augmented cubes \cite{Hung2015Augmented}, folded locally twisted cubes and folded crossed cubes \cite{Pai2023FLTQ}, balanced hypercubes \cite{Cheng2024Balanced}, and BCube data-center networks \cite{Pai2026BCube}.  Recent graph-theoretic work has also studied EDHCs in symmetric and Cayley-type graphs \cite{Yang2023Cayley}.  These works show that EDHC constructions remain active and relevant for both network theory and communication applications.

\section{Preliminaries}
\label{sec:preliminaries}

Let
\begin{equation}
    \mathbb{Z}[i]=\{x+yi:x,y\in\mathbb{Z}\}
\end{equation}
be the ring of Gaussian integers.  The norm of $\alpha=a+bi$ is
\begin{equation}
    N(\alpha)=a^2+b^2.
\end{equation}

\begin{definition}[Gaussian network]
Let $\alpha=a+bi\ne 0$.  The Gaussian network generated by $\alpha$ is the graph
\begin{equation}
    G_\alpha=(V,E),
\end{equation}
where
\begin{equation}
    V=\mathbb{Z}[i]/(\alpha),
\end{equation}
and two vertices $\beta,\gamma\in V$ are adjacent if and only if
\begin{equation}
    \beta-\gamma\equiv \pm 1 \quad \text{or} \quad \pm i \pmod{\alpha}.
\end{equation}
\end{definition}

The number of vertices is
\begin{equation}
    |V|=N=a^2+b^2.
\end{equation}
Throughout the paper, assume $0<a\leq b$.  Let
\begin{equation}
    d=\gcd(a,b),\qquad r=\frac{a^2+b^2}{d}.
\end{equation}

\subsection{Rectangle Representation}

The rectangle representation uses the complete residue system
\begin{equation}
    S=\{(x,y):0\leq x<r,\;0\leq y<d\}.
\end{equation}
Thus, each vertex is represented by a pair $(x,y)$, with $x$ taken modulo $r$ and $y$ restricted to $0,\ldots,d-1$.

Let $u,v\in\mathbb{Z}$ satisfy
\begin{equation}
    ua+vb=d.
\end{equation}
Define
\begin{equation}
    n\equiv av-bu \pmod r.
\end{equation}
Then the rectangle representation has the wrap relation
\begin{equation}
    -i\equiv n+(d-1)i \pmod{\alpha}.
\end{equation}
Consequently,
\begin{equation}
    di\equiv -n \pmod{\alpha}.
\end{equation}

The moves in rectangle coordinates are:
\begin{align}
    E(x,y)&=(x+1 \bmod r,y),\\
    W(x,y)&=(x-1 \bmod r,y),\\
    N_i(x,y)&=
    \begin{cases}
        (x,y+1), & 0\leq y<d-1,\\
        (x-n \bmod r,0), & y=d-1,
    \end{cases}\\
    S_i(x,y)&=
    \begin{cases}
        (x,y-1), & 0<y\leq d-1,\\
        (x+n \bmod r,d-1), & y=0.
    \end{cases}
\end{align}
Here $E,W,N_i,S_i$ denote east, west, north, and south moves.

\begin{remark}
It is important that the $x$-coordinate changes by $-n$ only when the $+i$ move wraps from row $d-1$ to row $0$.  It is not correct to subtract $n$ at every vertical step.
\end{remark}

\begin{lemma}
\label{lem:n-multiple-d}
The wrap shift $n$ is divisible by $d$.
\end{lemma}

\begin{proof}
Write $a=da_1$ and $b=db_1$.  Then
\begin{equation}
    n\equiv av-bu=d(a_1v-b_1u) \pmod r.
\end{equation}
Also,
\begin{equation}
    r=\frac{a^2+b^2}{d}=d(a_1^2+b_1^2),
\end{equation}
so $r$ is divisible by $d$.  Hence the chosen residue $n$ modulo $r$ is a multiple of $d$.
\end{proof}

\section{The Unified Switching Construction}
\label{sec:construction}

The construction begins with all horizontal edges.  These edges form $d$ disjoint row cycles, each of length $r$.  The construction then performs $d-1$ local switches.

For $d=1$, the switch set is empty.  The horizontal cycle and vertical cycle are already Hamiltonian, so the construction reduces to the known coprime case.

For $d>1$, define for every
\begin{equation}
    q=1,2,\ldots,d-1
\end{equation}
the switching column
\begin{equation}
    a_q=q-1.
\end{equation}
Define the four switch vertices
\begin{align}
    A_q&=(a_q,q)=(q-1,q),\\
    B_q&=(a_q+1,q)=(q,q),\\
    C_q&=A_q+i,\\
    D_q&=B_q+i.
\end{align}
The switch removes the two horizontal edges
\begin{equation}
    A_qB_q,\qquad C_qD_q,
\end{equation}
and inserts the two vertical edges
\begin{equation}
    A_qC_q,\qquad B_qD_q.
\end{equation}

Let $R$ be the set of removed horizontal edges and $S$ the set of inserted vertical edges:
\begin{align}
    R&=\{A_qB_q,C_qD_q:1\le q\le d-1\},\\
    S&=\{A_qC_q,B_qD_q:1\le q\le d-1\}.
\end{align}
Let $E_h$ be the set of all horizontal edges and $E_v$ the set of all vertical edges.  Define
\begin{equation}
    E(H_1)=(E_h\setminus R)\cup S,
\end{equation}
and
\begin{equation}
    E(H_2)=(E_v\setminus S)\cup R.
\end{equation}
Equivalently, $H_2$ is the edge-complement of $H_1$ inside $G_\alpha$, using precisely the graph edges not selected by $H_1$.  Intuitively, $G_\alpha$ is $4$-regular and $H_1$ uses exactly two incident edges at every vertex; the remaining two incident edges at each vertex form the complementary $2$-regular graph $H_2$.

Figure~\ref{fig:local-switch} illustrates the local switch operation.  The dashed horizontal edges are removed from $H_1$ and become part of $H_2$, while the solid vertical switch edges are inserted into $H_1$ and removed from $H_2$.

\begin{figure}[!t]
\centering
\includegraphics[width=0.55\linewidth]{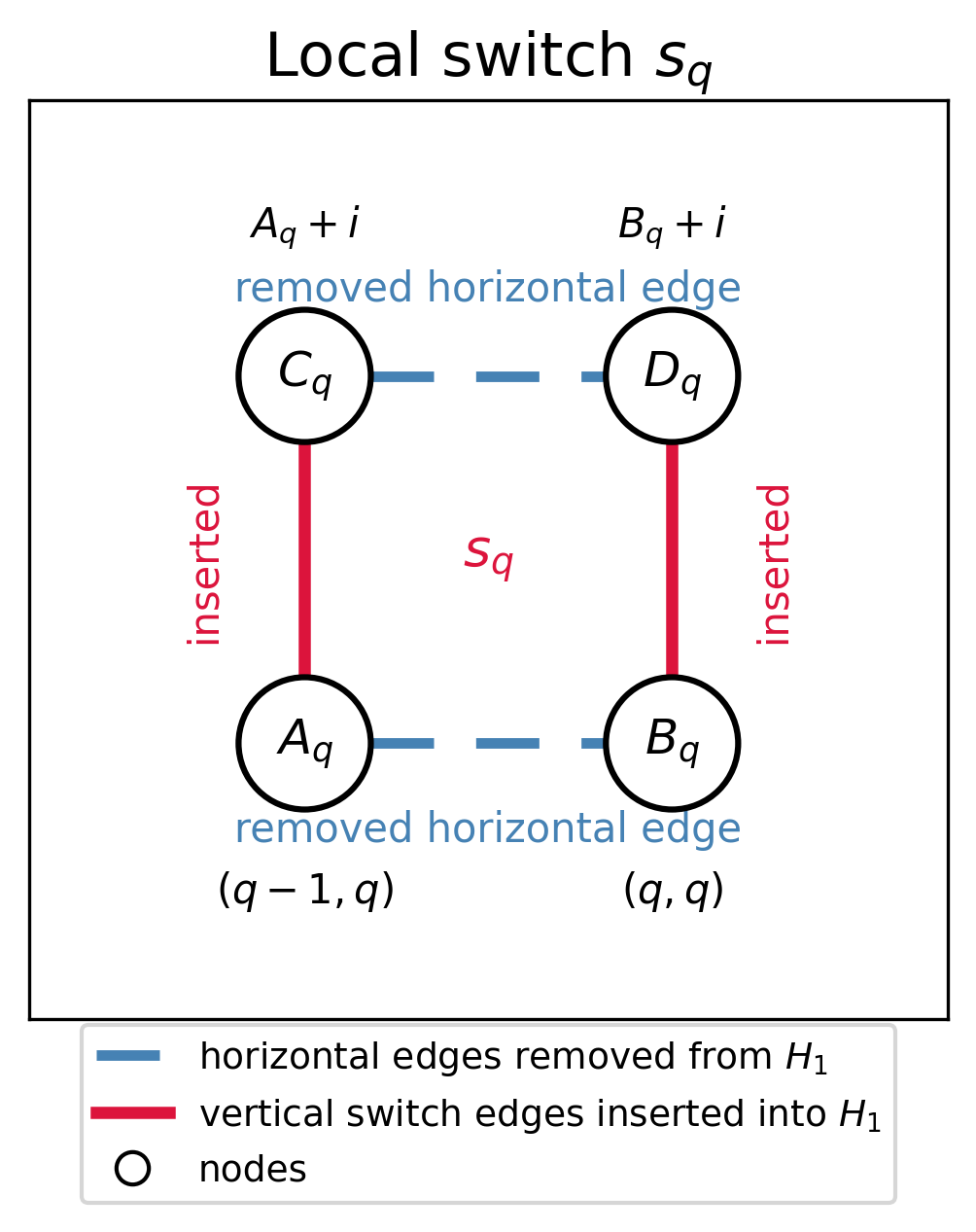}
\caption{Local switch $s_q$.  The two dashed horizontal edges $A_qB_q$ and $C_qD_q$ are removed from the horizontal construction of $H_1$ and are assigned to $H_2$.  The two solid vertical edges $A_qC_q$ and $B_qD_q$ are inserted into $H_1$ and excluded from $H_2$.}
\label{fig:local-switch}
\end{figure}

The global effect of the same switch family is summarized in Fig.~\ref{fig:concept}.  In the horizontal structure, the switches splice the $d$ row cycles into one Hamiltonian cycle $H_1$.  In the complementary vertical structure, the corresponding exchanged edges splice the vertical cycles into one Hamiltonian cycle $H_2$.

\begin{figure}[!t]
\centering
\includegraphics[width=\textwidth]{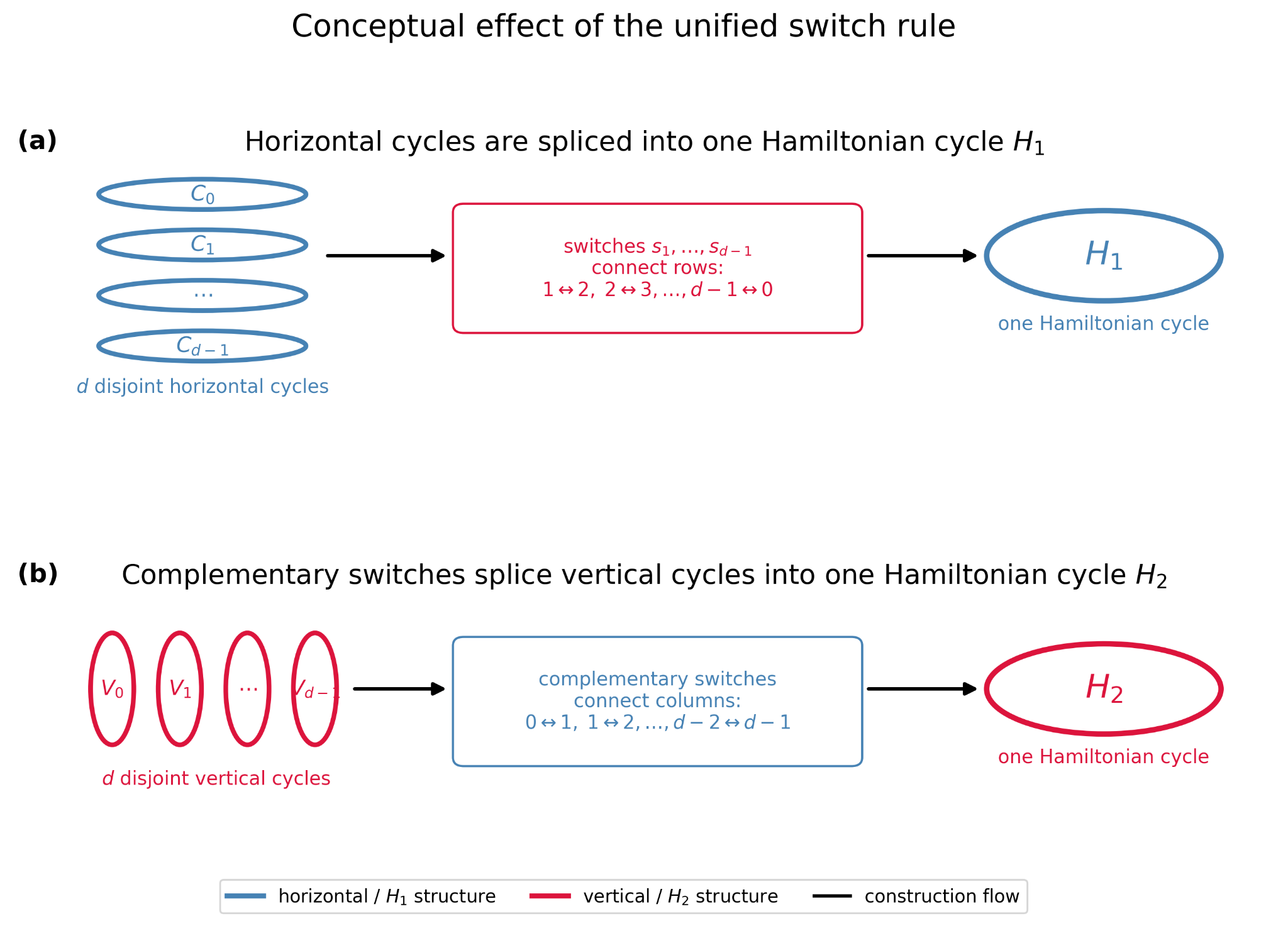}
\caption{Conceptual effect of the unified switch rule.  The switches $s_1,\ldots,s_{d-1}$ splice the horizontal row cycles into the Hamiltonian cycle $H_1$, while the complementary exchanged edges splice the vertical cycles into the Hamiltonian cycle $H_2$.}
\label{fig:concept}
\end{figure}

\section{Correctness Proof}
\label{sec:proof}

We first state a basic cycle-splicing fact.

\begin{lemma}[Two-edge splicing]
\label{lem:splicing}
Let $C_1$ and $C_2$ be two vertex-disjoint cycles.  Suppose $xy$ is an edge of $C_1$ and $uv$ is an edge of $C_2$.  Removing $xy$ and $uv$ and adding $xu$ and $yv$ produces one cycle containing all vertices of $C_1\cup C_2$, provided that the added edges are valid graph edges.
\end{lemma}

\begin{proof}
Removing one edge from each cycle turns $C_1$ and $C_2$ into two paths.  The two added edges connect the two paths end-to-end, producing one cycle.  No vertex is removed or repeated.
\end{proof}

\begin{lemma}
\label{lem:switches-valid-disjoint}
For $1\leq q\leq d-1$, each local switch is valid, and no two switches use the same removed or inserted edge.
\end{lemma}

\begin{proof}
By construction, $A_q$ and $B_q$ differ by $+1$, so $A_qB_q$ is horizontal.  Also, $C_q=A_q+i$ and $D_q=B_q+i$.  Therefore, $C_q$ and $D_q$ differ by $+1$, so $C_qD_q$ is horizontal.  The edges $A_qC_q$ and $B_qD_q$ are vertical by definition.

The removed edge $A_qB_q$ lies on row $q$, while $C_qD_q$ lies on row $q+1\pmod d$.  The formula $a_q=q-1$ gives distinct switch positions across the row adjacencies.  The inserted vertical edges have lower endpoints $A_q$ and $B_q$, which are distinct for different $q$.  Hence the inserted edges are distinct.  The only row that appears in two consecutive switches is handled at different horizontal columns, so the removed horizontal edges are also distinct.
\end{proof}

\begin{lemma}
\label{lem:h1-hamiltonian}
For $d>1$, the graph $H_1$ is a Hamiltonian cycle.
\end{lemma}

\begin{proof}
Before switching, $E_h$ consists of $d$ disjoint row cycles:
\begin{equation}
    C_0,C_1,\ldots,C_{d-1},
\end{equation}
where $C_y$ denotes the horizontal cycle in row $y$.

The switch at row $q$ removes one horizontal edge from row $q$ and one horizontal edge from row $q+1$ for $1\le q\le d-2$.  It then inserts two vertical edges between these two rows.  Therefore, for $1\le q\le d-2$, switch $s_q$ splices $C_q$ with $C_{q+1}$.  For the last switch, $q=d-1$, the vertices $C_{d-1}=A_{d-1}+i$ and $D_{d-1}=B_{d-1}+i$ are obtained through the vertical wrap from row $d-1$ to row $0$; hence $s_{d-1}$ splices $C_{d-1}$ with $C_0$.

Thus the row-cycle adjacencies created by the switches are
\begin{equation}
\begin{array}{c|c}
\text{switch} & \text{row cycles spliced}\\
\hline
s_1 & C_1 \leftrightarrow C_2\\
s_2 & C_2 \leftrightarrow C_3\\
\vdots & \vdots\\
s_{d-2} & C_{d-2} \leftrightarrow C_{d-1}\\
s_{d-1} & C_{d-1} \leftrightarrow C_0.
\end{array}
\end{equation}
These $d-1$ adjacencies form a connected tree on the $d$ row cycles.  Equivalently, the row cycles are connected in the order
\begin{equation}
    C_0-C_{d-1}-C_{d-2}-\cdots-C_2-C_1.
\end{equation}
By Lemma~\ref{lem:splicing}, each switch reduces the number of cycles by one.  Starting from $d$ cycles and applying $d-1$ switches leaves exactly one cycle.  Since the switches preserve degree two and keep all vertices, the resulting cycle contains all $N$ vertices.  Therefore $H_1$ is Hamiltonian.
\end{proof}

\begin{lemma}
\label{lem:vertical-index}
The vertical cycles before switching are indexed by $x\bmod d$.
\end{lemma}

\begin{proof}
A non-wrap vertical step changes $(x,y)$ to $(x,y+1)$ and therefore preserves $x\bmod d$.  A wrap vertical step changes $(x,d-1)$ to $(x-n,0)$.  By Lemma~\ref{lem:n-multiple-d}, $n$ is divisible by $d$, so $x-n\equiv x\pmod d$.  Hence every vertical step preserves $x\bmod d$, and each vertical component is contained in one residue class modulo $d$.

It remains to show that each residue class gives one cycle, not several smaller cycles.  Write
\begin{equation}
    a=da_1,\qquad b=db_1,\qquad r=d(a_1^2+b_1^2).
\end{equation}
Since $ua+vb=d$, we have
\begin{equation}
    ua_1+vb_1=1,
\end{equation}
so $\gcd(a_1,b_1)=1$.  Also write $n=dn'$ and $r=dr'$, where
\begin{equation}
    n'=a_1v-b_1u,\qquad r'=a_1^2+b_1^2.
\end{equation}
We claim that
\begin{equation}
    \gcd(n',r')=1.
\end{equation}
Indeed, if a prime $p$ divides both $n'=a_1v-b_1u$ and $r'=a_1^2+b_1^2$, then using $ua_1+vb_1=1$ gives
\begin{align}
    a_1 &= a_1(ua_1+vb_1)
        =u a_1^2+v a_1b_1
        \equiv -u b_1^2+v a_1b_1 \nonumber\\
        &=b_1(a_1v-b_1u)\equiv 0 \pmod p.
\end{align}
Similarly, $b_1\equiv0\pmod p$, contradicting $\gcd(a_1,b_1)=1$.  Hence $\gcd(n',r')=1$.

Now fix a residue class $c\bmod d$.  Its columns are
\begin{equation}
    c,\ c+d,\ c+2d,\ldots,c+(r'-1)d.
\end{equation}
After one full vertical traversal through the $d$ rows, the wrap changes $x$ by $-n=-dn'$, which is the step $-n'$ on the $r'$ columns inside the fixed residue class.  Since $\gcd(n',r')=1$, repeated full wraps visit all $r'$ columns in that residue class before returning.  Each such column contributes all $d$ rows, so the component has $dr'=r$ vertices.  Therefore each residue class $x\bmod d$ is exactly one vertical cycle of length $r$, and there are precisely $d$ such cycles.
\end{proof}

\begin{lemma}
\label{lem:h2-hamiltonian}
For $d>1$, the graph $H_2$ is a Hamiltonian cycle.
\end{lemma}

\begin{proof}
Before switching, $E_v$ consists of $d$ vertical cycles.  By Lemma~\ref{lem:vertical-index}, these cycles are indexed by $x\bmod d$.

Consider the switch for a fixed $q$, where $1\le q\le d-1$.  The horizontal edge
\begin{equation}
    A_qB_q=((q-1,q),(q,q))
\end{equation}
connects the vertical-cycle index $q-1$ to the vertical-cycle index $q$.  Since $C_q=A_q+i$ and $D_q=B_q+i$, adding $i$ preserves the $x$-coordinate modulo $d$ by Lemma~\ref{lem:vertical-index}.  Hence $C_qD_q$ also connects the vertical-cycle index $q-1$ to the vertical-cycle index $q$.

In $H_2$, the two vertical edges $A_qC_q$ and $B_qD_q$ are removed from the corresponding vertical cycles, while the two horizontal edges $A_qB_q$ and $C_qD_q$ are inserted.  By Lemma~\ref{lem:splicing}, this splices the two vertical cycles indexed by $q-1$ and $q$.

As $q$ ranges from $1$ to $d-1$, the complement switches join the vertical-cycle indices through
\begin{equation}
    0-1-2-\cdots-(d-1),
\end{equation}
which is a tree on the $d$ vertical cycles.  Therefore the $d-1$ complement switches reduce the $d$ vertical cycles to one cycle.  Since the switches preserve degree two and all vertices remain present, $H_2$ is Hamiltonian.
\end{proof}

\begin{lemma}
\label{lem:edge-disjoint-partition}
For every nondegenerate Gaussian network with $N>2$, the constructed graphs $H_1$ and $H_2$ are edge-disjoint and their edge sets partition $E(G_\alpha)$.
\end{lemma}

\begin{proof}
Let $E_h$ be the set of edges generated by $\pm1$, and let $E_v$ be the set of edges generated by $\pm i$.  Since $N>2$, no undirected edge can be both horizontal and vertical.  Otherwise, one would have $\pm1\equiv\pm i\pmod{\alpha}$, so $\alpha$ would divide one of $1+i$ or $1-i$, whose norm is $2$, contradicting $N(\alpha)=N>2$.  Hence
\begin{equation}
    E_h\cap E_v=\emptyset.
\end{equation}
Also,
\begin{equation}
    E(G_\alpha)=E_h\cup E_v.
\end{equation}

The switch construction chooses $R\subseteq E_h$ and $S\subseteq E_v$ and defines
\begin{equation}
    E(H_1)=(E_h\setminus R)\cup S,
\end{equation}
while
\begin{equation}
    E(H_2)=(E_v\setminus S)\cup R.
\end{equation}
Thus every horizontal edge is placed in exactly one of $H_1$ and $H_2$: the edges in $R$ go to $H_2$, and the edges in $E_h\setminus R$ remain in $H_1$.  Similarly, every vertical edge is placed in exactly one of $H_1$ and $H_2$: the edges in $S$ go to $H_1$, and the edges in $E_v\setminus S$ remain in $H_2$.  Therefore
\begin{equation}
    E(H_1)\cap E(H_2)=\emptyset
\end{equation}
and
\begin{equation}
    E(H_1)\cup E(H_2)=E(G_\alpha).
\end{equation}
\end{proof}

\begin{theorem}
\label{thm:main}
Let $\alpha=a+bi$ with $0<a\leq b$, let $N=a^2+b^2>2$, and let $G_\alpha$ be the corresponding nondegenerate Gaussian network.  The construction defined by
\begin{equation}
    q=1,2,\ldots,d-1,\qquad a_q=q-1,
\end{equation}
where $d=\gcd(a,b)$, constructs two edge-disjoint Hamiltonian cycles in $G_\alpha$.  It applies uniformly to $d=1$, $d>1$, odd $d$, and even $d$.
\end{theorem}

\begin{proof}
If $d=1$, no switch is used.  The horizontal edges form one Hamiltonian cycle and the vertical edges form a second Hamiltonian cycle, as in the known coprime case \cite{FlahiveBose2010}.  These cycles are edge-disjoint because they use different dimensions.

Now assume $d>1$.  Lemma~\ref{lem:h1-hamiltonian} shows that $H_1$ is Hamiltonian.  Lemma~\ref{lem:h2-hamiltonian} shows that $H_2$ is Hamiltonian.  Lemma~\ref{lem:edge-disjoint-partition} shows that their edge sets are disjoint and partition the edge set of $G_\alpha$:
\begin{equation}
    E(H_1)\cap E(H_2)=\emptyset,
    \qquad
    E(H_1)\cup E(H_2)=E(G_\alpha).
\end{equation}
Thus, $H_1$ and $H_2$ are two edge-disjoint Hamiltonian cycles.  The proof does not depend on whether $d$ is odd or even.
\end{proof}

\begin{corollary}
\label{cor:edge-decomposition}
For $d>1$, the two cycles have the following edge-type counts:
\begin{align}
    |E(H_1)\cap E_h| &= N-2(d-1),\\
    |E(H_1)\cap E_v| &= 2(d-1),\\
    |E(H_2)\cap E_h| &= 2(d-1),\\
    |E(H_2)\cap E_v| &= N-2(d-1).
\end{align}
\end{corollary}

\begin{proof}
Each of the $d-1$ switches removes two horizontal edges from $H_1$ and inserts two vertical edges into $H_1$.  Thus $H_1$ contains $N-2(d-1)$ horizontal edges and $2(d-1)$ vertical edges.  Since $H_2$ is the complement, it contains the opposite counts.
\end{proof}

\section{Algorithms}
\label{sec:algorithms}

This section gives explicit algorithms for constructing the two Hamiltonian cycles.  Algorithm~\ref{alg:switchset} constructs the switch sets.  Algorithm~\ref{alg:neighbors} gives constant-time membership for the two cycles.  Algorithm~\ref{alg:traverse} lists the Hamiltonian-cycle order.

\begin{paperalgorithm}{ConstructSwitchSets$(a,b)$}{alg:switchset}
\REQUIRE Generator $\alpha=a+bi$, $0<a\le b$
\ENSURE Removed set $R$ and inserted set $S$
\STATE $d\leftarrow \gcd(a,b)$; $N\leftarrow a^2+b^2$; $r\leftarrow N/d$
\STATE Find $u,v$ such that $ua+vb=d$
\STATE $n\leftarrow av-bu \pmod r$
\STATE $R\leftarrow\emptyset$, $S\leftarrow\emptyset$
\FOR{$q=1$ to $d-1$}
    \STATE $a_q\leftarrow q-1$
    \STATE $A\leftarrow(a_q,q)$; $B\leftarrow(a_q+1,q)$
    \STATE $C\leftarrow A+i$ using the rectangle move
    \STATE $D\leftarrow B+i$ using the rectangle move
    \STATE $R\leftarrow R\cup\{AB,CD\}$
    \STATE $S\leftarrow S\cup\{AC,BD\}$
\ENDFOR
\RETURN $R,S$
\end{paperalgorithm}

\begin{paperalgorithm}{CycleNeighbors$(v,R,S)$}{alg:neighbors}
\REQUIRE Vertex $v$, switch sets $R,S$
\ENSURE The two neighbors of $v$ in $H_1$ and in $H_2$
\STATE $N_G(v)\leftarrow\{v\pm1,v\pm i\}$
\STATE $N_1\leftarrow\emptyset$; $N_2\leftarrow\emptyset$
\FOR{each $w\in N_G(v)$}
    \STATE $e\leftarrow\{v,w\}$
    \IF{$e\in S$ or ($e$ is horizontal and $e\notin R$)}
        \STATE $N_1\leftarrow N_1\cup\{w\}$
    \ELSE
        \STATE $N_2\leftarrow N_2\cup\{w\}$
    \ENDIF
\ENDFOR
\RETURN $N_1,N_2$
\end{paperalgorithm}

\begin{paperalgorithm}{TraverseCycle$(s,N_C)$}{alg:traverse}
\REQUIRE Start vertex $s$, neighbor function $N_C(\cdot)$ for $H_1$ or $H_2$
\ENSURE Vertex order of the Hamiltonian cycle
\STATE $L\leftarrow [s]$; $\text{prev}\leftarrow \bot$; $\text{cur}\leftarrow s$
\REPEAT
    \STATE Let $N_C(\text{cur})=\{z_1,z_2\}$
    \IF{$z_1\ne \text{prev}$}
        \STATE $\text{next}\leftarrow z_1$
    \ELSE
        \STATE $\text{next}\leftarrow z_2$
    \ENDIF
    \STATE $\text{prev}\leftarrow \text{cur}$; $\text{cur}\leftarrow \text{next}$
    \IF{$\text{cur}\ne s$}
        \STATE append $\text{cur}$ to $L$
    \ENDIF
\UNTIL{$\text{cur}=s$}
\RETURN $L$
\end{paperalgorithm}

\subsection{Complexity}

The switch rule computes each switch using a constant number of arithmetic operations.  Therefore, constructing the switch sets takes $O(d)$ time.  Since each Hamiltonian cycle has $N$ edges, listing either cycle takes $O(N)$ time.  The total cycle-generation time is therefore $O(N)$, and the local rule for each switch is constant time.

\section{Examples}
\label{sec:examples}

\subsection{Example 1: The Coprime Case $\alpha=3+5i$}

Here
\begin{equation}
    d=\gcd(3,5)=1,\qquad N=34.
\end{equation}
The switch set is empty.  The first Hamiltonian cycle is obtained by repeatedly adding $1$, and the second is obtained by repeatedly adding $i$.  Thus, the construction reduces exactly to the previously known $d=1$ case.

\subsection{Example 2: The Non-Coprime Case $\alpha=4+4i$}

Here
\begin{equation}
    d=4,\qquad N=32,\qquad r=8.
\end{equation}
The switches are
\begin{equation}
    (q,a_q)=(1,0),(2,1),(3,2).
\end{equation}
The first two switches connect consecutive rows without wrap-around:
\begin{align}
    (0,1)(1,1),\;(0,2)(1,2),\\
    (1,2)(2,2),\;(1,3)(2,3).
\end{align}
The final switch uses the wrap from row $3$ to row $0$:
\begin{equation}
    A_3=(2,3),\quad B_3=(3,3),
\end{equation}
and because $n=4$,
\begin{equation}
    C_3=(6,0),\quad D_3=(7,0).
\end{equation}
This produces
\begin{equation}
    H_1:26\text{ horizontal edges}+6\text{ vertical edges},
\end{equation}
and
\begin{equation}
    H_2:6\text{ horizontal edges}+26\text{ vertical edges}.
\end{equation}

The switch data for this example are
\begin{equation}
\begin{array}{c|c|c|c}
q & a_q & A_q,B_q & C_q,D_q\\
\hline
1 & 0 & (0,1),(1,1) & (0,2),(1,2)\\
2 & 1 & (1,2),(2,2) & (1,3),(2,3)\\
3 & 2 & (2,3),(3,3) & (6,0),(7,0).
\end{array}
\end{equation}
The last switch is the wrap switch from row $3$ to row $0$.

\begin{figure}[H]
	\centering
	\begin{minipage}{0.49\textwidth}
		\centering
		\includegraphics[width=\linewidth]{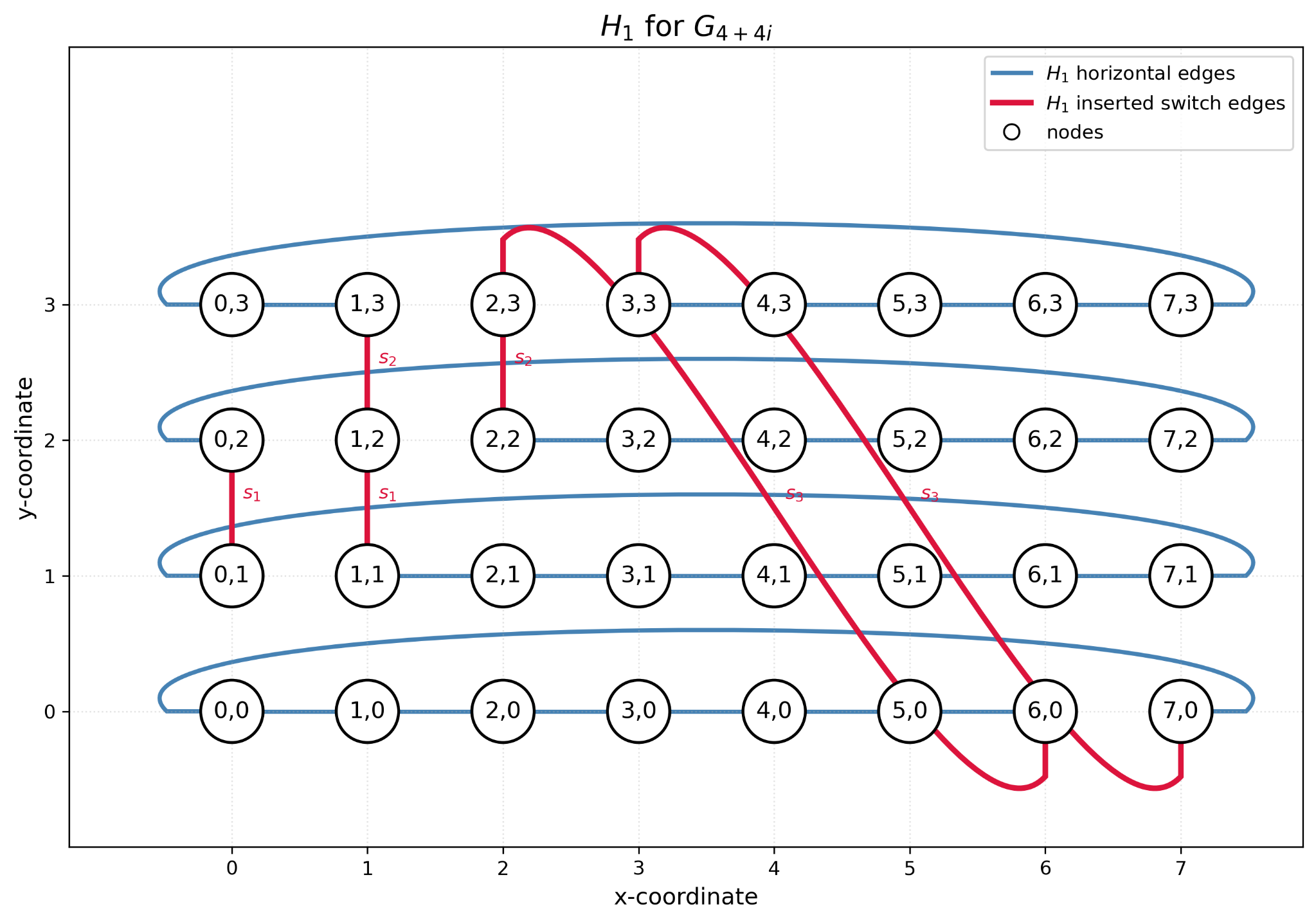}\\[-1mm]
		\textbf{(a)} $H_1$
	\end{minipage}
	\hfill
	\begin{minipage}{0.49\textwidth}
		\centering
		\includegraphics[width=\linewidth]{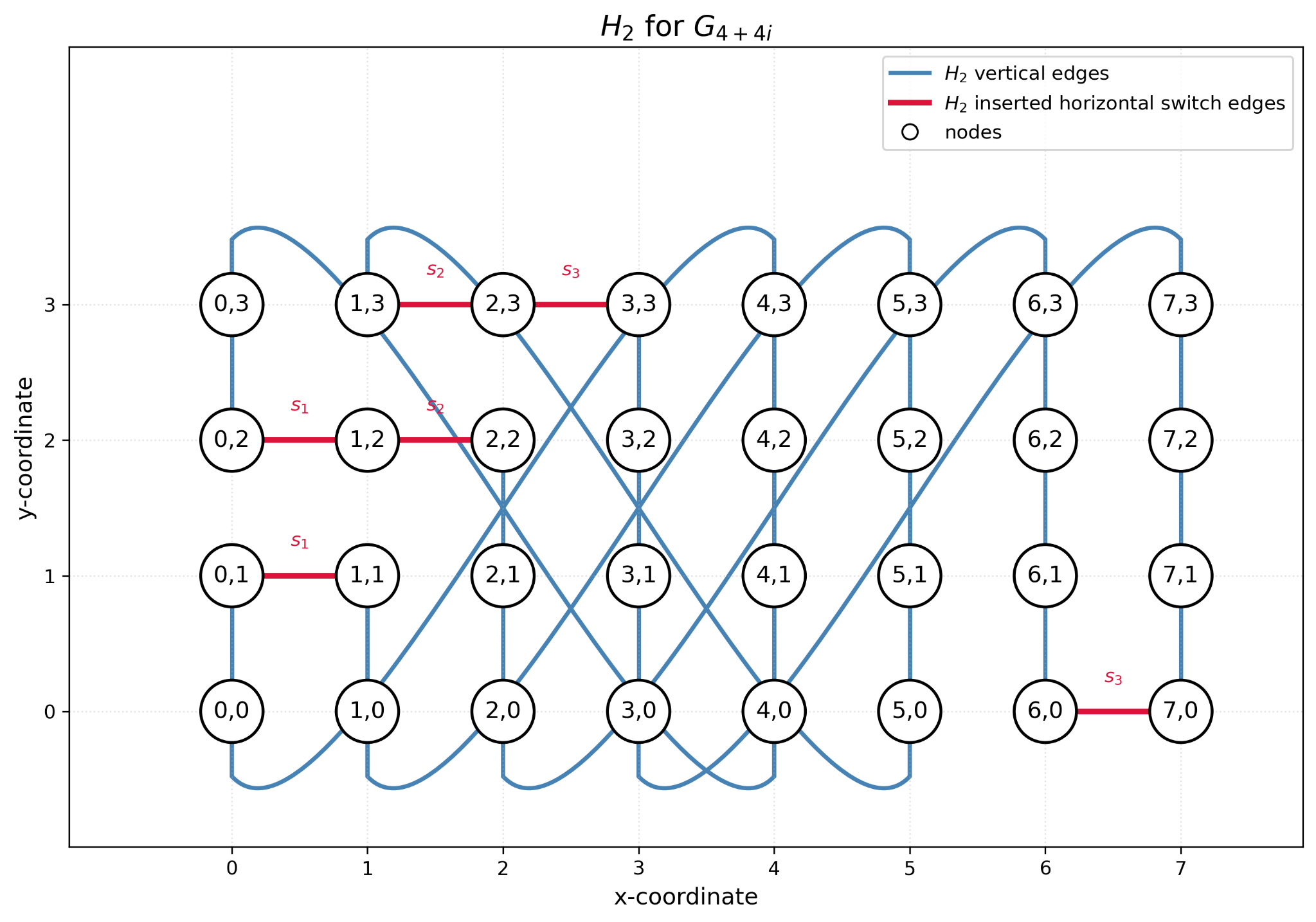}\\[-1mm]
		\textbf{(b)} $H_2$
	\end{minipage}
	\caption{The two Hamiltonian cycles for $G_{4+4i}$ shown separately.  (a) The cycle $H_1$ is obtained from the horizontal row-cycle structure by applying the switch operations; the red edges are the inserted vertical switch edges.  (b) The cycle $H_2$ is obtained from the complementary vertical structure; the red edges are the horizontal switch edges removed from $H_1$ and assigned to $H_2$.}
	\label{fig:example-4-4i-separate}
\end{figure}

\begin{figure}[H]
	\centering
	\includegraphics[width=0.95\textwidth]{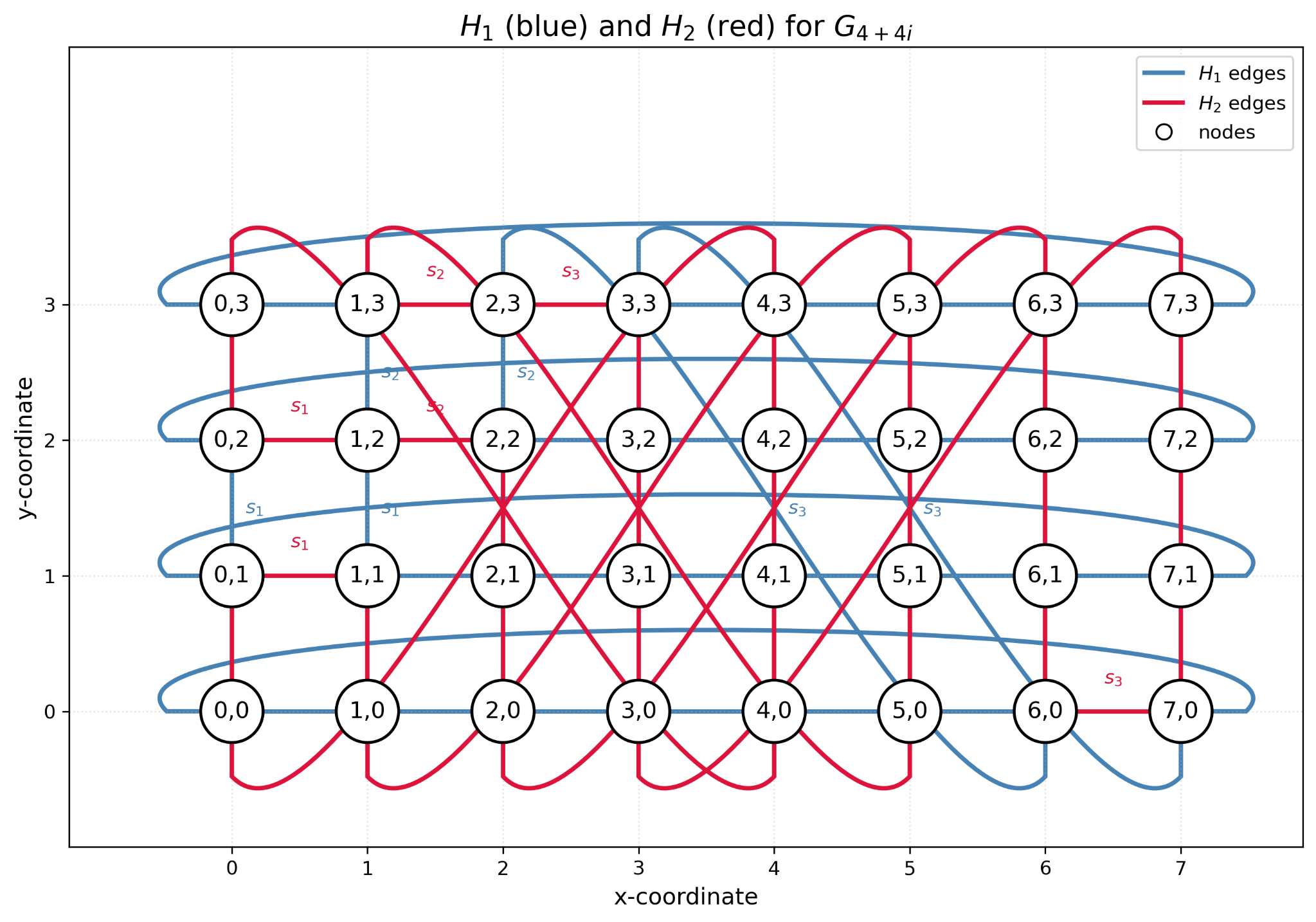}
	\caption{Combined drawing of the two edge-disjoint Hamiltonian cycles in $G_{4+4i}$.  The cycle $H_1$ is drawn in blue and the cycle $H_2$ is drawn in red.  The two cycles have no common graph edge, i.e., $E(H_1)\cap E(H_2)=\emptyset$.  Crossings in the drawing are visual crossings only and do not represent shared edges.}
	\label{fig:example-4-4i-combined}
\end{figure}

Figures~\ref{fig:example-4-4i-separate} and~\ref{fig:example-4-4i-combined} illustrate the complete construction for $G_{4+4i}$.  Figure~\ref{fig:example-4-4i-separate} shows the two Hamiltonian cycles separately: $H_1$ uses the switched horizontal structure, while $H_2$ uses the complementary vertical structure.  Figure~\ref{fig:example-4-4i-combined} draws the two cycles together, with $H_1$ in blue and $H_2$ in red.  The combined drawing emphasizes edge-disjointness: any crossings are geometric crossings in the drawing only and are not shared graph edges.

\subsection{Example 3: An Odd Non-Coprime Case $\alpha=5+10i$}

For $\alpha=5+10i$,
\begin{equation}
    d=5,\qquad N=125,\qquad r=25.
\end{equation}
The switch sequence is
\begin{equation}
    (q,a_q)=(1,0),(2,1),(3,2),(4,3).
\end{equation}
The validation gives
\begin{equation}
    |H_1|=|H_2|=125,
\end{equation}
and the edge decomposition is
\begin{equation}
    H_1:117\text{ horizontal}+8\text{ vertical},
\end{equation}
\begin{equation}
    H_2:8\text{ horizontal}+117\text{ vertical}.
\end{equation}

\section{Experimental Validation}
\label{sec:experiments}

A Python validation program was written to test the construction.  For each generator $\alpha=a+bi$, the program computes $d$, $r$, $n$, constructs the switch sets, verifies the degree of every vertex in $H_1$ and $H_2$, traverses each cycle, and confirms that each traversal visits all $N$ vertices exactly once.  It also records edge-type counts and runtime.  The degenerate network $\alpha=1+i$ with $N=2$ is excluded because the four formal directions collapse to a single undirected edge.

The validator checks four correctness conditions for every nondegenerate input.  First, every vertex has degree two in $H_1$ and degree two in $H_2$.  Second, a traversal of $H_1$ returns one cycle containing all $N$ vertices, and the same condition is checked independently for $H_2$.  Third, the edge sets are tested for disjointness, i.e., $E(H_1)\cap E(H_2)=\emptyset$.  Finally, the measured horizontal and vertical edge counts are compared with the closed-form counts in Corollary~\ref{cor:edge-decomposition}.

\subsection{Validation Scope}

The validation included:
\begin{itemize}
    \item 18 representative cases covering $d=1$, $d>1$, odd $d$, even $d$, small networks, and large networks;
    \item exhaustive validation for all $1\leq a\leq b\leq 100$;
    \item large-scale validation up to $N=3{,}250{,}000$.
\end{itemize}

\begin{table}[H]
\caption{Representative validation of the closed-form EDHC construction.}
\label{tab:representative-validation}
\centering
\footnotesize
\begin{tabular}{c|r|r|r|r|r|c}
\hline
$\alpha$ & $d$ & $r$ & $N$ & $|H_1|$ & $|H_2|$ & Result\\
\hline
$3+5i$ & 1 & 34 & 34 & 34 & 34 & Yes\\
$4+4i$ & 4 & 8 & 32 & 32 & 32 & Yes\\
$3+6i$ & 3 & 15 & 45 & 45 & 45 & Yes\\
$5+10i$ & 5 & 25 & 125 & 125 & 125 & Yes\\
$6+6i$ & 6 & 12 & 72 & 72 & 72 & Yes\\
$6+8i$ & 2 & 50 & 100 & 100 & 100 & Yes\\
$8+8i$ & 8 & 16 & 128 & 128 & 128 & Yes\\
$8+12i$ & 4 & 52 & 208 & 208 & 208 & Yes\\
$10+15i$ & 5 & 65 & 325 & 325 & 325 & Yes\\
$12+12i$ & 12 & 24 & 288 & 288 & 288 & Yes\\
$15+20i$ & 5 & 125 & 625 & 625 & 625 & Yes\\
$20+30i$ & 10 & 130 & 1300 & 1300 & 1300 & Yes\\
$25+35i$ & 5 & 370 & 1850 & 1850 & 1850 & Yes\\
$50+50i$ & 50 & 100 & 5000 & 5000 & 5000 & Yes\\
$60+80i$ & 20 & 500 & 10000 & 10000 & 10000 & Yes\\
$100+100i$ & 100 & 200 & 20000 & 20000 & 20000 & Yes\\
$100+150i$ & 50 & 650 & 32500 & 32500 & 32500 & Yes\\
$200+300i$ & 100 & 1300 & 130000 & 130000 & 130000 & Yes\\
\hline
\end{tabular}
\end{table}

\begin{table}[H]
\caption{Edge-type decomposition of the two Hamiltonian cycles.}
\label{tab:edge-decomposition}
\centering
\footnotesize
\begin{tabular}{c|r|r|r|r|r|r}
\hline
$\alpha$ & $d$ & $H_1^h$ & $H_1^v$ & $H_2^h$ & $H_2^v$ & $2(d-1)$\\
\hline
$3+5i$ & 1 & 34 & 0 & 0 & 34 & 0\\
$4+4i$ & 4 & 26 & 6 & 6 & 26 & 6\\
$3+6i$ & 3 & 41 & 4 & 4 & 41 & 4\\
$5+10i$ & 5 & 117 & 8 & 8 & 117 & 8\\
$6+6i$ & 6 & 62 & 10 & 10 & 62 & 10\\
$6+8i$ & 2 & 98 & 2 & 2 & 98 & 2\\
$8+8i$ & 8 & 114 & 14 & 14 & 114 & 14\\
$8+12i$ & 4 & 202 & 6 & 6 & 202 & 6\\
$10+15i$ & 5 & 317 & 8 & 8 & 317 & 8\\
$12+12i$ & 12 & 266 & 22 & 22 & 266 & 22\\
$20+30i$ & 10 & 1282 & 18 & 18 & 1282 & 18\\
$50+50i$ & 50 & 4902 & 98 & 98 & 4902 & 98\\
$100+100i$ & 100 & 19802 & 198 & 198 & 19802 & 198\\
$200+300i$ & 100 & 129802 & 198 & 198 & 129802 & 198\\
\hline
\end{tabular}
\end{table}

\begin{table}[H]
\caption{Aggregate correctness summary.}
\label{tab:aggregate-summary}
\centering
\footnotesize
\begin{tabular}{l|r|r|r|r}
\hline
Category & Cases & Checked & Successes & Failures\\
\hline
All & 5076 & 5075 & 5075 & 0\\
Exhaustive $1\leq a\leq b\leq 100$ & 5050 & 5049 & 5049 & 0\\
Large scale & 8 & 8 & 8 & 0\\
Representative & 18 & 18 & 18 & 0\\
$d=1$ & 3045 & 3044 & 3044 & 0\\
$d>1$ & 2031 & 2031 & 2031 & 0\\
Odd $d$ & 3781 & 3780 & 3780 & 0\\
Even $d$ & 1295 & 1295 & 1295 & 0\\
\hline
\end{tabular}
\end{table}

\subsection{Interpretation of the Experiments}

The exhaustive test contains $5050$ generator pairs with $1\leq a\leq b\leq 100$.  The only excluded pair is $\alpha=1+i$, where $N=2$.  Therefore, $5049$ nondegenerate exhaustive cases were checked and all passed.  The validation also separately confirms the important categories $d=1$, $d>1$, odd $d$, and even $d$.  The large-scale tests show linear behavior, because the validation time grows approximately proportionally with $N$.

\begin{table}[H]
\caption{Runtime scalability of the validation procedure.}
\label{tab:runtime-scalability}
\centering
\footnotesize
\begin{tabular}{c|r|r|r|r}
\hline
$\alpha$ & $d$ & $N$ & Total ms & $\mu$s/node\\
\hline
$100+100i$ & 100 & 20000 & 521.907 & 26.095350\\
$100+150i$ & 50 & 32500 & 892.941 & 27.475108\\
$200+200i$ & 200 & 80000 & 2223.681 & 27.796012\\
$200+300i$ & 100 & 130000 & 5310.403 & 40.849254\\
$500+500i$ & 500 & 500000 & 16605.924 & 33.211848\\
$500+750i$ & 250 & 812500 & 26680.321 & 32.837318\\
$1000+1000i$ & 1000 & 2000000 & 68336.990 & 34.168495\\
$1000+1500i$ & 500 & 3250000 & 112296.254 & 34.552694\\
\hline
\end{tabular}
\end{table}

\section{Discussion}
\label{sec:discussion}

The construction has two useful interpretations.  First, it is a unifying theorem: the same rule covers both the old coprime case and the non-coprime case.  Second, it is an algorithmic proof: the switches are not found by search and are not described by long tables.  They are given by one formula.

The key design choice is to omit the row adjacency $0\leftrightarrow1$ from the horizontal-cycle splicing and instead include the wrap adjacency $(d-1)\leftrightarrow0$.  This makes $H_1$ join rows through
\begin{equation}
    0-(d-1)-(d-2)-\cdots-1,
\end{equation}
while the complement $H_2$ joins vertical cycle indices through
\begin{equation}
    0-1-2-\cdots-(d-1).
\end{equation}
Thus, one set of switches creates two different spanning splice trees: one for the horizontal cycles and one for the vertical cycles.

\section{Conclusion}
\label{sec:conclusion}

This paper presented a unified constant-time switch rule for constructing two edge-disjoint Hamiltonian cycles in Gaussian networks.  The construction covers the known coprime case $d=1$ and the non-coprime case $d>1$ in a single framework.  It also removes the need for separate odd/even $d$ cases.  The rule is simply to switch row adjacency $q$ at column $a_q=q-1$ for $q=1,\ldots,d-1$.  The switched horizontal graph is one Hamiltonian cycle, and its edge-complement is the second Hamiltonian cycle; the two cycles are edge-disjoint by construction.

The construction is compact, implementable, and linear-time for explicit cycle generation.  Exhaustive and large-scale validation confirmed the result for all tested nondegenerate Gaussian networks, including all generators with $1\leq a\leq b\leq 100$ and large networks up to $3{,}250{,}000$ vertices.

\section*{Funding}
This research received no external funding.

\section*{Data Availability}
The validation data and source code generated for this study are available from the corresponding author upon reasonable request and can be provided as supplementary material during peer review.

\section*{Conflicts of Interest}
The author declares no conflicts of interest.

\section*{Acknowledgments}
The author would like to acknowledge the support of Kuwait University and its Computer Science Department and would like to thank the researchers whose previous work on Gaussian networks and edge-disjoint Hamiltonian cycles motivated this unified construction.

\end{document}